\documentclass[12pt,letterpaper]{article}
\usepackage{amsmath,amsfonts,amsthm,amssymb,stmaryrd,relsize}

\theoremstyle{plain}

\numberwithin{equation}{section}
\newtheorem{thm}{Theorem}[section]

{\begin{flushleft}\textbf{Example #1}.\enspace}%
{\end{flushleft}}

\newcommand{\complex}{{\mathbb C}}
\newcommand{\real}{{\mathbb R}}
\newcommand{\ascript}{{\mathcal A}}
\newcommand{\cscript}{{\mathcal C}}
\newcommand{\dscript}{{\mathcal D}}
\newcommand{\pscript}{{\mathcal P}}
\newcommand{\rscript}{{\mathcal R}}
\newcommand{\qscript}{{\mathcal Q}}
\newcommand{\tscript}{{\mathcal T}}
\newcommand{\rmim}{\mathrm{Im}}
\newcommand{\rmcyl}{\mathrm{cyl}}
\newcommand{\ftilde}{\widetilde{f}}
\newcommand{\ctimes}{\mathrel{\mathlarger\cdot}}

\newcommand{\ab}[1]{\left|#1\right|}
\newcommand{\brac}[1]{\left\{#1\right\}}
\newcommand{\paren}[1]{\left(#1\right)}
\newcommand{\sqbrac}[1]{\left[#1\right]}
\newcommand{\elbows}[1]{{\left\langle#1\right\rangle}}
\newcommand{\ket}[1]{{\left|#1\right>}}
\newcommand{\bra}[1]{{\left<#1\right|}}

\begin{document}

\title{AN EINSTEIN EQUATION FOR\\DISCRETE QUANTUM GRAVITY
}
\author{S. Gudder\\ Department of Mathematics\\
University of Denver\\ Denver, Colorado 80208, U.S.A.\\
sgudder@du.edu
}
\date{}
\maketitle

\begin{abstract}
The basic framework for this article is the causal set approach to discrete quantum gravity (DQG). Let $Q_n$ be the collection of causal sets with cardinality not greater than $n$ and let $K_n$ be the standard Hilbert space of complex-valued functions on $Q_n$. The formalism of DQG presents us with a decoherence matrix $D_n(x,y)$,
$x,y\in Q_n$. There is a growth order in $Q_n$ and a path in $Q_n$ is a maximal chain relative to this order. We denote the set of paths in $Q_n$ by $\Omega _n$. For $\omega ,\omega '\in\Omega _n$ we define a bidifference operator
$\varbigtriangledown _{\omega ,\omega '}^n$ on $K_n\otimes K_n$ that is covariant in the sense that
$\varbigtriangledown _{\omega ,\omega '}^n$ leaves $D_n$ stationary. We then define the curvature operator
$\rscript _{\omega ,\omega '}^n=\varbigtriangledown _{\omega ,\omega '}^n
  -\varbigtriangledown _{\omega ',\omega}^n$. It turns out that $\rscript _{\omega ,\omega '}^n$ naturally decomposes into two parts $\rscript _{\omega ,\omega '}^n=\dscript _{\omega ,\omega '}^n+\tscript _{\omega ,\omega '}^n$ where
$\dscript _{\omega ,\omega '}^n$ is closely associated with $D_n$ and is called the metric operator while
$\tscript _{\omega ,\omega '}^n$ is called the mass-energy operator. This decomposition is a discrete analogue of Einstein's equation of general relativity. Our analogue may be useful in determining whether general relativity theory is a close approximation to DQG.
\end{abstract}

\section{Causet Approach to DQG}  
A \textit{causal set} (\textit{causet}) is a finite partially ordered set $x$. Thus, $x$ is endowed with an irreflexive, transitive relation $<$ \cite{blms87, sor03,sur11}. That is, $a\not< a$ for all $a\in x$ and $a<b$, $b<c$ imply that $a<c$ for
$a,b,c\in x$. The relation $a<b$ indicates that $b$ is in the causal future of $a$. Let $\pscript _n$ be the collection of all causets with cardinality $n$, $n=1,2,\ldots$, and let $\pscript =\cup\pscript _n$. For $x\in\pscript$, an element $a\in x$ is \textit{maximal} if there is no $b\in x$ with $a<b$. If $x\in\pscript _n$, $y\in\pscript _{n+1}$, then $x$ \textit{produces} $y$ if $y$ is obtained from $x$ by adjoining a single new element $a$ to $x$ that is maximal in $y$. In this way, there is no element of
$y$ in the causal future of $a$. If $x$ produces $y$, we say that $y$ is an \textit{offspring} of $x$ and write $x\to y$.

A \textit{path} in $\pscript$ is a string (sequence) $x_1x_2\cdots$ where $x_i\in\pscript _i$ and $x_i\to x_{i+1}$,
$i=1,2,\ldots\,$. An $n$-\textit{path} in $\pscript$ is a finite string $x_1x_2\cdots x_n$ where again $x_i\in\pscript _i$ and
$x_i\to x_{i+1}$. We denote the set of paths by $\Omega$ and the set of $n$-paths by $\Omega _n$. We think of
$\omega\in\Omega$ as a possible universe (or  universe history). The set of paths whose initial $n$-path is
$\omega _0\in\Omega _n$ is called an \textit{elementary cylinder set} and is denoted by $\rmcyl (\omega _0)$. Thus, if
$\omega _0=x_1x_2\cdots x_n$, then
\begin{equation*}
\rmcyl (\omega _0)=\brac{\omega\in\Omega\colon\omega =x_1x_2\cdots x_ny_{n+1}y_{n+2}\cdots}
\end{equation*}
The \textit{cylinder set generated} by $A\subseteq\Omega _n$ is defined by
\begin{equation*}
\rmcyl (A)=\bigcup _{\omega\in A}\rmcyl (\omega )
\end{equation*}
The collection $\ascript _n=\brac{\rmcyl (A)\colon A\subseteq\Omega _n}$ forms an increasing sequence of algebras $\ascript _1\subseteq\ascript _2\subseteq\cdots$ on $\Omega$ and hence $\cscript (\Omega )=\cup\ascript _n$ is an algebra of subsets of $\Omega$. We denote the $\sigma$-algebra generated by $\cscript (\Omega )$ as $\ascript$.

It is shown \cite{rs00,vr06} that a classical sequential growth process (CSGP) on $\pscript$ that satisfies natural causality and covariance conditions is determined by a sequence of nonnegative numbers $c=(c_0,c_1,\ldots)$ called
\textit{coupling constants}. The coupling constants specify a unique probability measure $\nu _c$ on $\ascript$ making
$(\Omega ,\ascript ,\nu _c)$ a probability space. The \textit{path Hilbert space} is given by
$H=L_2(\Omega ,\ascript ,\nu _c)$. If $\nu _c^n=\nu _c\mid\ascript _n$ is the restriction of $\nu _c$ to $\ascript _n$, then
$H_n=L_2(\Omega ,\ascript _n,\nu _c^n)$ is an increasing sequence of closed subspaces of $H$.

A bounded operator $T$ on $H_n$ will also be considered as a bounded operator on $H$ by defining $Tf=0$ for all
$f\in H_n^\perp$. We denote the characteristic function of a set $A\in\ascript$ by $\chi _A$ and use the notation
$\chi _\Omega =1$. A $q$-\textit{probability operator} is a bounded positive operator $\rho _n$ on $H_n$ that satisfies
$\elbows{\rho _n1,1}=1$. Denote the set of $q$-probability operators on $H_n$ by $\qscript (H_n)$. For
$\rho _n\in\qscript (H_n)$ we define the $n$-\textit{decoherence functional} \cite{gud111,hen09,sor03}
$D_n\colon\ascript _n\times\ascript _n\to\complex$ by
\begin{equation*}
D_n(A,B)=\elbows{\rho _n\chi _B,\chi _A}
\end{equation*}
The functional $D_n(A,B)$ gives a measure of the interference between the events $A$ and $B$ when the system is described by $\rho _n$. It is clear that $D_n(\Omega _n,\Omega _n)=1$, $D_n(A,B)=\overline{D_n(B,A)}$ and
$A\mapsto D_n(A,B)$ is a complex measure for all $B\in\ascript _n$. It is also well known that if
$A_1,\ldots ,A_n\in\ascript _n$, then the matrix with entries $D_n(A_j,A_k)$ is positive semidefinite. In particular the positive semidefinite matrix with entries
\begin{equation*}
D_n(\omega _i,\omega _j)=D_n\paren{\rmcyl (\omega _i),\rmcyl (\omega _j)},\quad\omega _i,\omega _j\in\Omega _n
\end{equation*}
is called the $n$-\textit{decoherence matrix}.

We define the map $\mu _n\colon\ascript _n\to\real ^+$ given by
\begin{equation*}
\mu _n(A)=D_n(A,A)=\elbows{\rho _n\chi _A,\chi _A}
\end{equation*}
Notice that $\mu _n(\Omega )=1$. Although $\mu _n$ is not additive in general, it satisfies the \textit{grade}-2
\textit{additivity condition} \cite{djs10,gud111,hen09,sor94,sor11}: if $A,B,C\in\ascript _n$ are mutually disjoint  then
\begin{equation*}
\mu _n(A\cup B\cup C)=\mu _n(A\cup B)+\mu _n(A\cup C)+\mu _n(B\cup C)-\mu _n(A)-\mu _n(B)-\mu _n(C)
\end{equation*}
We call $\mu _n$ the $q$-\textit{measure} corresponding to $\rho _n$ and interpret $\mu _n(A)$ as the quantum propensity for the occurrence of the event $A\in\ascript _n$. A simple example is to let $\rho _n=I$, $n=1,2,\ldots\,$. Then
\begin{equation*}
D_n(A,B)=\elbows{\chi _B,\chi _A}=\nu _c^n(A\cap B)
\end{equation*}
and $\mu _n(A)=\nu _c^n(A)$, the classical measure. A more interesting example is to let $\rho _n=\ket{1}\bra{1}$, $n=1,2,\ldots\,$. Then
\begin{equation*}
D_n(A,B)=\elbows{\ket{1}\bra{1}\chi _B,\chi _A}=\nu _c^n(A)\nu _c^n(B)
\end{equation*}
and $\mu _n(A)=\sqbrac{\nu _c^n(A)}^2$, the classical measure square.

We call a sequence $\rho _n\in\qscript (H_n)$, $n=1,2,\ldots$, \textit{consistent} if
\begin{equation*}
D_{n+1}(A,B)=D_n(A,B)
\end{equation*}
for all $A,B\in\ascript _n$. A \textit{quantum sequential growth process} (QSGP) is a consistent sequence
$\rho _n\in\qscript (H_n)$ \cite{gud111,gud112} . We consider a QSGP as a model for discrete quantum gravity. It is hoped that additional theoretical principles or experimental data will help determine the coupling constants and hence $\nu _c$, which is the classical part of the process, and also the $\rho _n\in\qscript (H_n)$, which is the quantum part. Moreover, it is believed that general relativity will eventually be shown to be a close approximation to this discrete model. Until now it has not been clear how this can be accomplished. However, in Section~3 we shall derive a discrete Einstein equation which might be useful in performing these tasks.

\section{Difference Operators} 

Let $Q_n={\displaystyle\bigcup _{j=1}^n}\pscript _j$ be the collection of causets with cardinality not greater than $n$ and let $K_n$ be the (finite-dimensional) Hilbert space $\complex ^{Q_n}$ with the standard inner product
\begin{equation*}
\elbows{f,g}=\sum _{x\in Q_n}\overline{f(x)}g(x)
\end{equation*}
Let $L_n=K_n\otimes K_n$ which we identify with $\complex ^{Q_n\times Q_n}$ having the standard inner product. We shall also have need to consider the Hilbert space
\begin{equation*}
K=\brac{f\in\complex ^{\pscript}\colon\sum _{x\in\pscript}\ab{f(x)}^2<\infty}
\end{equation*}
with the standard inner product and we define $L=K\otimes K$. Notice that $K_1\subseteq K_2\subseteq\cdots K$ form an increasing sequence of subspaces of $K$ that generate $K$ in the natural way.

Let $\rho _n\in\qscript (H_n)$ be a QSGP with corresponding decoherence matrix $D_n(\omega ,\omega ')$,
$\omega ,\omega '\in\Omega _n$. If $\omega =\omega _1\omega _2\cdots\omega _n\in \Omega _n$ and $\omega _j=x$ for some $j$, we say that $\omega$ \textit{goes through} $x$. For $x,y\in Q_n$ we define
\begin{equation*}
D_n(x,y)=\sum\brac{D_n(\omega ,\omega ')\colon\omega\hbox{ goes through } x,\ \omega '\hbox{ goes through }y}
\end{equation*}
Due to the consistency of $\rho _n$, $D_n(x,y)$ is independent of $n$. That is, $D_n(x,y)\!=\!D_m(x,y)$ if
$x,y\in Q_n\cap Q_m$. Moreover, $D_n(x,y)$ are the components of a positive semidefinite matrix. We view $Q_n$ as the analogue of a differentiable manifold and $D_n(x,y)$ as the analogue of a metric tensor. One might think that the elements of causets should be analogous to points of a differential manifold and not the causets themselves. However, if $x\in Q_n$, then $x$ is intimately related to its producers, each of which determines a unique $a\in x$. Moreover, if $y\to x$ we view
$(y,x)$ as a tangent vector at $x$. In this way, there are as many tangent vectors at $x$ as there are producers of $x$. Finally, the elements of $\Omega _n$ are analogues of curves and the elements of $K_n$ are analogues of smooth functions on a manifold.

For $x\in Q_n$, $\ab{x}$ denotes the cardinality of $x$. Notice if
$\omega =\omega _1\omega _2\cdots\omega _n\in\Omega _n$ and $\omega _j=x$, then $j=\ab{x}$ and $\omega$ goes through $x$ if and only if $\omega _{\ab{x}}=x$. We see that a path $\omega$ through $x$ determines a tangent vector
$(\omega _{\ab{x}-1},\omega _{\ab{x}})$ at $x$ (assuming that $\ab{x}\ge 2$). For $\omega\in\Omega _n$ we define the \textit{difference operator} $\varbigtriangleup _\omega ^n$ on $K_n$ by
\begin{equation*}
\varbigtriangleup _\omega ^nf(x)=\sqbrac{f(x)-f(\omega _{\ab{x}-1})}\delta _{x,\omega _{\ab{x}}}
\end{equation*}
for all $f\in K_n$, where $\delta _{x,\omega _{\ab{x}}}$ is the Kronecker delta. Thus, $\varbigtriangleup _\omega ^nf(x)$ gives the change of $f$ along the tangent vector $(\omega _{\ab{x}-1},\omega _{\ab{x}})$ if $\omega$ goes through $x$. It is clear that $\varbigtriangleup _\omega ^n$ is a linear operator on $K_n$. We now show that $\varbigtriangleup _\omega ^n$ satisfies a discrete form of Leibnitz's rule. For $f,g\in K_n$ we have
\begin{align}             
\label{eq21}
\varbigtriangleup _\omega ^nfg(x)
  &=\sqbrac{f(x)g(x)-f(\omega _{\ab{x}-1})g(\omega _{\ab{x}-1})}\delta _{x,\omega _{\ab{x}}}\notag\\
  &=\brac{\sqbrac{f(x)g(x)\!-\!f(x)g(\omega _{\ab{x}-1})}
  \!\!+\!\!\sqbrac{f(x)g(\omega _{\ab{x}-1})\!-\!f(\omega _{\ab{x}-1})g(\omega _{\ab{x}-1})}}\notag\\
  &\qquad\ctimes\delta _{x,\omega _{\ab{x}}}\notag\\
  &=f(x)\varbigtriangleup _\omega ^ng(x)+g(\omega _{\ab{x}-1})\varbigtriangleup _\omega ^nf(x)
\end{align}
Of course, it also follows that
\begin{equation}         
\label{eq22}
\varbigtriangleup _\omega ^nfg(x)=\varbigtriangleup _\omega ^ngf(x)
  =g(x)\varbigtriangleup _\omega ^nf(x)+f(\omega _{\ab{x}-1})\varbigtriangleup _\omega ^ng(x)
\end{equation}
Given a function of two variables $f\in\complex ^{Q_n\times Q_n}=L_n$ we have a function $\ftilde\in K_n$ of one variable where $\ftilde (x)=f(x,x)$ and given a function $f\in K_n$ we have the functions of two variables $f_1,f_2\in L_n$ where
$f_1(x,y)=f(x)$ and $f_2(x,y)=f(y)$ for all $x,y\in Q_n$. For $\omega ,\omega '\in\Omega _n$ we want an operator
$\varbigtriangleup _{\omega ,\omega '}^n\colon L_n\to L_n$ that extends $\varbigtriangleup _\omega ^n$ and satisfies a discrete Leibnitz's rule. That is,
\begin{align}           
\label{eq23}
\varbigtriangleup _{\omega ,\omega '}^nf_1(x,y)
  &=\varbigtriangleup _\omega ^nf(x)\delta _{y,\omega '_{\ab{x}}},
  \varbigtriangleup _{\omega ,\omega '}^nf_2(x,y)
  =\varbigtriangleup _{\omega '}^nf(y)\delta _{x,\omega _{\ab{x}}}\\
\intertext{and}
\label{eq24}
\varbigtriangleup _{\omega ,\omega '}^nfg(x,y)
  &=f(x,y)\varbigtriangleup _{\omega ,\omega '}g(x,y)+g(\omega _{\ab{x}-1},\omega '_{\ab{x}-1})
  \varbigtriangleup _{\omega ,\omega '}^nf(x,y)
\end{align}

\begin{thm}       
\label{thm21}
A linear operator $\varbigtriangleup _{\omega ,\omega '}^n\colon L_n\to L_n$ satisfies \eqref{eq23} and \eqref{eq24} if and only if $\varbigtriangleup _{\omega ,\omega '}^n$ has the form
\begin{equation}         
\label{eq25}
\varbigtriangleup _{\omega ,\omega '}^nf(x,y)=\sqbrac{f(x,y)-f(\omega _{\ab{x}-1},\omega '_{\ab{y}-1})}
  \delta _{x,\omega _{\ab{x}}}\delta _{y,\omega '_{\ab{y}}}
\end{equation}
\end{thm}
\begin{proof}
If $\varbigtriangleup _{\omega ,\omega '}^n$ is defined by \eqref{eq25}, then for $f\in K_n$ we have
\begin{align*}
\varbigtriangleup _{\omega ,\omega '}^nf_1(x,y)&=\sqbrac{f_1(x,y)-f_1(\omega _{\ab{x}-1},\omega '_{\ab{y}-1})}
  \delta _{x,\omega _{\ab{x}}}\delta _{y,\omega '_{\ab{y}}}\\
  &=\sqbrac{f(x)-f(\omega _{\ab{x}-1})}\delta _{x,\omega _{\ab{x}}}\delta _{y,\omega '_{\ab{y}}}\\
  &=\varbigtriangleup _\omega ^nf(x)\delta _{y,\omega '_{\ab{y}}}
\end{align*}
In a similar way, $\varbigtriangleup _{\omega ,\omega '}^n$ satisfies the second equation in \eqref{eq23}. Moreover, we have
\begin{align*}
\varbigtriangleup _{\omega ,\omega '}^nf(x,y)
  &=\sqbrac{f(x,y)g(x,y)-f(\omega _{\ab{x}-1}\omega '_{\ab{y}-1})g(\omega _{\ab{x}-1}\omega '_{\ab{y}-1})}
  \delta _{x,\omega _{\ab{x}}}\delta _{y,\omega '_{\ab{y}}}\\
  &=\sqbrac{f(x,y)g(x,y)-f(x,y)g(\omega _{\ab{x}-1},\omega '_{\ab{y}-1})}
  \delta _{x,\omega _{\ab{x}}}\delta _{y,\omega '_{\ab{y}}}\\
  &\quad +\sqbrac{f(x,y)g(\omega _{\ab{x}-1},\omega '_{\ab{y}-1})-f(\omega _{\ab{x}-1}\omega '_{\ab{y}-1})
  g(\omega _{\ab{x}-1}\omega '_{\ab{y}-1})}\\
  &\qquad\ctimes\delta _{x,\omega _{\ab{x}}}\delta _{y,\omega '{\ab{y}}}\\
  &=f(x,y)\varbigtriangleup _{\omega ,\omega '}^ng(x,y)+g(\omega _{\ab{x}-1},\omega '_{\ab{y}-1})
  \varbigtriangleup _{\omega ,\omega '}^nf(x,y)
\end{align*}
Conversely, suppose the linear operator $\varbigtriangleup _{\omega ,\omega '}^n\colon L_n\to L_n$ satisfies \eqref{eq23} and \eqref{eq24}. If $f\in L_n$ has the form $f(x,y)=g(x)h(y)$, then
\begin{align*}
\varbigtriangleup _{\omega ,\omega '}^nf(x,y)
  &=\varbigtriangleup _{\omega ,\omega '}^ngh(x,y)=\varbigtriangleup _{\omega ,\omega '}^ng_1h_2(x,y)\\
  &=g_1(x,y)\varbigtriangleup _{\omega ,\omega '}^n
  h_2(x,y)+h_2(\omega _{\ab{x}-1}\omega '_{\ab{y}-1})\varbigtriangleup _{\omega ,\omega '}^ng_1(x,y)\\
  &=g(x)\varbigtriangleup _{\omega '}^nh(y)\delta _{x,\omega _{\ab{x}}}+h(\omega '_{\ab{y}-1})
  \varbigtriangleup _\omega ^ng(x)\delta _{y,\omega '_{\ab{y}}}\\
  &=\brac{g(x)\sqbrac{h(y)-h(\omega '_{\ab{y}-1})}+h(\omega '_{\ab{y}-1})\sqbrac{g(x)-g(\omega _{\ab{y}-1})}}\\
  &\qquad\ctimes\delta _{x,\omega _{\ab{x}}}\delta _{y,\omega '{\ab{y}}}\\
  &=\sqbrac{g(x)h(y)-g(\omega _{\ab{x}-1})h(\omega '_{\ab{y}-1})}\delta _{x,\omega _{\ab{x}}}\delta _{y,\omega '_{\ab{y}}}\\
  &=\sqbrac{f(x,y)-f(\omega _{\ab{x}-1}\omega '_{\ab{y}-1})}\delta _{x,\omega _{\ab{x}}}\delta _{y,\omega '_{\ab{y}}}
\end{align*}
Since $\varbigtriangleup _{\omega ,\omega '}^n$, is linear and every element of $L_n$ is a linear combination of product functions, the result follows.
\end{proof}

Of course, Theorem~\ref{thm21} is not surprising because \eqref{eq25} is the natural extension of
$\varbigtriangleup _\omega ^n$ to $L_n$. Also $\varbigtriangleup _{\omega ,\omega '}^n$ extends
$\varbigtriangleup _\omega ^n$ in the sense that for any $f\in L_n$ we have 
\begin{align*}
\varbigtriangleup _{\omega ,\omega}^nf(x,x)
  &=\sqbrac{f(x,y)-f(\omega _{\ab{x}-1},\omega _{\ab{x}-1})}\delta _{x,\omega _{\ab{x}}}\\
  &=\sqbrac{\ftilde (x)-\ftilde (\omega _{\ab{x}-1})}\delta _{x,\omega _{\ab{x}}}=\varbigtriangleup _\omega ^n\ftilde (x)
\end{align*}
As before, $\varbigtriangleup _{\omega ,\omega '}^n$ satisfies
\begin{align*}
\varbigtriangleup _{\omega ,\omega '}^nfg(x,y)&=\varbigtriangleup _{\omega ,\omega '}^ngf(x,y)\\
  &=g(x,y)\varbigtriangleup _{\omega ,\omega '}^nf(x,y)+f(\omega _{\ab{x}-1},\omega '_{\ab{y}-1})
  \varbigtriangleup _{\omega ,\omega '}^ng(x,y)
\end{align*}
The next result characterizes $\varbigtriangleup _\omega ^n$ and $\varbigtriangleup _{\omega ,\omega '}^n$ up to a multiplicative constant.

\begin{thm}       
\label{thm22}
{\rm (a)}\enspace An operator $T_\omega\colon K_n\to K_n$ satisfies \eqref{eq21} and $T_\omega f(x)=0$ if
$\omega _{\ab{x}}\ne x$ if and only if there exists a function $\beta _\omega\in K_n$ such that
$T_\omega =\beta _\omega\varbigtriangleup _\omega ^n$.
{\rm (b)}\enspace An operator $T_{\omega ,\omega '}\colon L_n\to L_n$ satisfies \eqref{eq24} and
$T_{\omega ,\omega '}f(x,y)=0$ if $\omega _{\ab{x}}\ne x$ or $\omega '_{\ab{y}}\ne y$ if and only if there exists a function
$\beta _{\omega ,\omega '}\in L_n$ such that
$T_{\omega ,\omega '}=\beta _{\omega ,\omega '}\varbigtriangleup _{\omega ,\omega '}^n$.
\end{thm}
\begin{proof}
If $T_\omega$ satisfies \eqref{eq21}, it follows from \eqref{eq22} that
\begin{equation*}
f(x)T_\omega g(x)+g(\omega _{\ab{x}-1})T_\omega f(x)=g(x)T_\omega f(\omega )+f(\omega _{\ab{x}-1})T_\omega g(x)
\end{equation*}
Hence,
\begin{equation*}
\sqbrac{g(x)-g(\omega _{\ab{x}-1})}T_\omega f(x)=T_\omega g(x)\sqbrac{f(x)-f(\omega _{\ab{x}-1})}
\end{equation*}
Therefore, if $g(x)-g(\omega _{\ab{x}-1})\ne 0$, we have
\begin{equation*}
T_\omega f(x)=\frac{T_\omega g(x)}{g(x)-g(\omega _{\ab{x}-1})}\sqbrac{f(x)-f(\omega _{\ab{x}-1})}
\end{equation*}
Letting
\begin{equation*}
\beta _\omega (x)=\frac{T_\omega g(x)}{g(x)-g(\omega _{\ab{x}-1})}
\end{equation*}
gives the result. The converse is straightforward. The proof of (b) is similar.
\end{proof}

It is clear that $\mu _n(x)=D_n(x,x)$ is not stationary; that is $\varbigtriangleup _\omega ^n\mu _n(x)\ne 0$ for all $x\in Q_n$ in general. Is there a function $\alpha _\omega\in K_n$ such that
$(\varbigtriangleup _\omega ^n+\alpha _\omega )\mu _n(x)=0$ for all $x\in Q_n$? If $\alpha _\omega$ exists, we obtain
\begin{equation*}
\sqbrac{\mu _n(x)-\mu _n(\omega _{\ab{x}-1})}\delta _{x,\omega _{\ab{x}}}+\alpha _\omega (x)\mu _n(x)=0
\end{equation*}
If $\omega _{\ab{x}}=x$ and $\mu _n(x)=0$, this would imply that $\mu _n(\omega _{\ab{x}-1})=0$. Continuing this process would give
\begin{equation*}
\mu _n(\omega _{\ab{x}-2})=\mu (\omega _{\ab{x}-3})=\cdots =0
\end{equation*}
which leads to a contradiction. It is entirely possible for $\mu _n(x)$ to be zero for some $x\in Q_n$ so we abandon this attempt. How about functions $\alpha _\omega ,\beta _\omega\in K_n$ such that
$(\beta _\omega\varbigtriangleup _\omega ^n+\alpha _\omega )\mu _n(x)=0$ for all $x\in Q_n$? We then obtain
\begin{equation}         
\label{eq26}
\beta _\omega (x)\sqbrac{\mu _n(x)-\mu _n(\omega _{\ab{x}-1})}
  \delta _{x,\omega _{\ab{x}}}+\alpha _\omega (x)\mu _n(x)=0
\end{equation}
If $\omega _{\ab{x}}=x$ and $\mu _n(x)=0$ but $\beta _\omega (x)\ne 0$ we obtain the same contradiction as before. We conclude that $\beta _\omega (x)=0$ whenever $\mu _n(x)=0$. The simplest choice of such a $\beta _\omega$ is
$\beta _\omega (x)=\mu _n(x)$. This choice also has the advantage of being independent of $\omega$.
Equation~\eqref{eq26} becomes
\begin{equation}         
\label{eq27}
\mu _n(x)\sqbrac{\mu _n(x)-\mu _n(\omega _{\ab{x}-1})}\delta _{x,\omega _{\ab{x}}}+\alpha _\omega (x)\mu _n(x)=0
\end{equation}
If $\mu _n(x)=0$, then \eqref{eq27} holds. If $\mu _n(x)\ne 0$, then we obtain
\begin{equation*} 
\alpha _\omega (x)=\sqbrac{\mu _n(\omega _{\ab{x}-1})-\mu _n(x)}
  \delta _{x,\omega _{\ab{x}}}=-\varbigtriangleup _\omega ^n\mu _n(x)
\end{equation*}
The numbers $\alpha _\omega (x)$ are an analogue of the Christoffel symbols. We call
$\varbigtriangledown _\omega ^n=\mu _n\varbigtriangleup _\omega ^n+\alpha _\omega$ the 
\textit{covariant difference operator}. The operator $\varbigtriangledown _\omega ^n$ is not a difference operator in the usual sense because $\varbigtriangledown _\omega ^n1\ne 0$. Instead, we have
$\varbigtriangledown _\omega ^n1=\alpha _\omega$.

Following the previous steps for $\varbigtriangleup _{\omega ,\omega '}^n$ we define the
\textit{covariant bidifference operator}
$\varbigtriangledown _{\omega ,\omega '}^n=D_n\varbigtriangleup _{\omega ,\omega '}^n+\alpha _{\omega ,\omega '}$ where
\begin{equation*} 
\alpha _{\omega ,\omega '}(x,y)=\sqbrac{D_n(\omega _{\ab{x}-1},\omega '_{\ab{y}-1})-D_n(x,y)}
  \delta _{x,\omega _{\ab{x}}}\delta _{y,\omega '_{\ab{y}}}
\end{equation*}
and again, $\alpha _{\omega ,\omega '}(x,y)$ are analogous to Christoffel symbols. Notice that
$\varbigtriangledown _{\omega ,\omega '}^nD_n(x,y)=0$ for all $x,y\in Q_n$ and
$\varbigtriangledown _{\omega ,\omega}^nf(x,x)=\varbigtriangledown _\omega ^n\ftilde (x)$. Complete expressions for
$\varbigtriangledown _\omega ^n$ and $\varbigtriangledown _{\omega ,\omega '}^n$ are
\begin{align}           
\label{eq28}
\varbigtriangledown _\omega ^nf(x)&=\sqbrac{\mu _n(\omega _{\ab{x}-1})f(x)-\mu _n(x)f(\omega _{\ab{x}-1})}
  \delta _{x,\omega _{\ab{x}}}\\
\intertext{and}
\label{eq29}
\varbigtriangledown _{\omega ,\omega '}^nf(x,y)&=\sqbrac{D_n(\omega _{\ab{x}-1},\omega '_{\ab{y}-1})f(x,y)
  -D_n(x,y)f(\omega _{\ab{x}-1},\omega '_{\ab{y}-1})}\notag\\
  &\qquad\ctimes\delta_{x,\omega _{\ab{x}}}\delta _{y,\omega '_{\ab{y}}}
\end{align}
The form of \eqref{eq28} and \eqref{eq29} shows that $\varbigtriangledown _\omega ^n$ and
$\varbigtriangledown _{\omega ,\omega '}^n$ are ``weighted'' difference operators.

\section{Curvature Operators} 
The linear operator $\rscript _{\omega ,\omega '}^n\colon L_n\to L_n$ defined by
\begin{equation*}
\rscript _{\omega ,\omega '}^n=\varbigtriangledown _{\omega ,\omega '}^n-\varbigtriangledown _{\omega ',\omega}^n
\end{equation*}
is called the \textit{curvature operator}. Applying \eqref{eq29} we have
\begin{align}         
\label{eq31}
\rscript&_{\omega ,\omega '}^nf(x,y)\notag\\
  &=D_n(x,y)\sqbrac{f(\omega '_{\ab{x}-1},\omega _{\ab{y}-1})
  \delta _{x,\omega '_{\ab{x}}}\delta _{y,\omega _{\ab{y}}}-f(\omega _{\ab{x}-1},\omega '_{\ab{y}-1})
   \delta _{x,\omega _{\ab{x}}}\delta _{y,\omega '_{\ab{y}}}}\notag\\
   &\quad +\sqbrac{D_n(\omega _{\ab{x}-1},\omega '_{\ab{y}-1})\delta _{x,\omega _{\ab{x}}}\delta _{y,\omega '_{\ab{y}}}
   -D_n(\omega '_{\ab{x}-1},\omega _{\ab{y}-1})\delta _{x,\omega '_{\ab{x}}}\delta _{y,\omega _{\ab{y}}}}f(x,y)
\end{align}
If $x=y$, then \eqref{eq31} reduces to
\begin{align*}
\rscript _{\omega ,\omega '}^nf(x,x)&=\mu _n(x)\sqbrac{f(\omega '_{\ab{x}-1},\omega _{\ab{x}-1})
  -f(\omega _{\ab{x}-1},\omega '_{\ab{x}-1})}\delta _{x,\omega _{\ab{x}}}\\
  &\quad +2i\,\rmim D_n(\omega _{\ab{x}-1},\omega '_{\ab{x}-1})f(x,x)\delta _{x,\omega _{\ab{x}}}
\end{align*}
We call the operator $\dscript _{\omega ,\omega '}^n\colon L_n\to L_n$ given by
\begin{align*}
\dscript &_{\omega ,\omega '}^nf(x,y)\\
  &=D_n(x,y)\sqbrac{f(\omega '_{\ab{x}-1},\omega _{\ab{y}-1})\delta _{x,\omega '_{\ab{x}}}\delta _{y,\omega _{\ab{y}}}
  -f(\omega _{\ab{x}-1},\omega '_{\ab{y}-1})\delta _{x,\omega _{\ab{x}}}\delta _{y,\omega '_{\ab{y}}}}
\end{align*}
the \textit{metric operator} and the operator $\tscript _{\omega ,\omega '}^n\colon L_n\to L_n$ given by
\begin{align*}
&\tscript _{\omega ,\omega '}^nf(x,y)\\
  &=\sqbrac{D_n(\omega _{\ab{x}-1},\omega '_{\ab{y}-1})\delta _{x,\omega _{\ab{x}}}\delta _{y,\omega '_{\ab{y}}}
  -D_n(\omega '_{\ab{x}-1},\omega _{\ab{y}-1})\delta _{x,\omega _{\ab{x}}}\delta _{y,\omega '_{\ab{y}}}}f(x,y)
\end{align*}
the \textit{mass-energy operator}. Then \eqref{eq31} gives
\begin{equation}         
\label{eq32}
\rscript _{\omega ,\omega '}^n=\dscript _{\omega ,\omega '}^n+\tscript _{\omega ,\omega '}^n
\end{equation}

Equation \eqref{eq32} is a discrete analogue of Einstein's equation \cite{wal84}. In this sense, Einstein's equation always holds in the present framework no matter what we have for the quantum dynamics $\rho _n$. One might argue that we obtained this discrete analogue of Einstein's equation just by definition. However, $\rscript _{\omega ,\omega '}^n$ is a reasonable counterpart of the curvature tensor in general relativity \cite{wal84} and $\dscript _{\omega ,\omega '}^n$ is certainly a counterpart of the metric tensor.

Equation \eqref{eq32} does not give direct information about $D_n(x,y)$ and\newline
$D_n(\omega ,\omega ')$ (which are, after all, what we want to find), but it may give useful indirect information. If we can find $D_n(\omega ,\omega ')$ such that the classical Einstein equation is an approximation to \eqref{eq32}, then this would give information about
$D_n(\omega ,\omega ')$. Moreover, an important problem in discrete quantum gravity theory is how to test whether general relativity is a close approximation to the theory. Whether Einstein's equation is an approximation to \eqref{eq32} would provide such a test. Another variant of a discrete Einstein equation can be obtained by defining the operator
$\rscript _{x,y}^n$ for $x,y\in Q_n$ by
\begin{equation*} 
\rscript _{x,y}^n=\sum\brac{\rscript _{\omega ,\omega '}^n\colon\omega _{\ab{x}}=x,\omega '_{\ab{y}}=y}
\end{equation*}
With similar definitions for $\dscript _{x,y}^n$ and $\tscript _{x,y}^n$ we obtain
\begin{equation*} 
\rscript _{x,y}^n=\dscript _{x,y}^n+\tscript _{x,y}^n
\end{equation*}
In order to consider approximations by Einstein's equation, it may be necessary to let $n\to\infty$ in \eqref{eq32}. However, the convergence of the operators depends on $D_n$ and will be left to a later paper. In a similar vein, it may be possible that limit operators $\rscript _{\omega ,\omega '}$, $\dscript _{\omega ,\omega '}$ and $\tscript _{\omega ,\omega '}$ can be defined as (possibly unbounded) operators directly on the Hilbert space $L$.

\section{Matrix Elements} 
We have introduced several operators on $K_n$ and $L_n$ in Sections~2 and 3. In order to understand such operators more directly, it is frequently useful to write them in terms of their matrix elements. First we denote the standard basis on
$K_n$ by $e_x^n,x\in Q_n$. The matrix that is zero except for a one in the $xy$ entry is denoted by
$\ket{e_x^n}\bra{e_y^n}$ and we call this the $xy$ \textit{matrix element}, $x,y\in Q_n$. Of course, in Dirac notation,
$\ket{e_x^n}\bra{e_y^n}$ can be considered directly as a linear operator without referring to a matrix. In any case, every linear operator $T$ on $K_n$ can be represented uniquely as
\begin{equation*} 
T=\sum _{x,y\in Q_n}t_{x,y}\ket{e_x^n}\bra{e_y^n}
\end{equation*}
for $t_{x,y}\in\complex$. In a similar way, $e_x^n\otimes e_y^n$, $x,y\in Q_n$ form an orthonormal basis for
$L_n=K_n\otimes K_n$ and every linear operator $T$ on $L_n$ has a unique representation
\begin{equation*} 
T=\sum\brac{t_{x,y;x',y'}\ket{e_x^n\otimes e_y^n}\bra{e_{x'}^n\otimes e_{y'}^n}\colon x,y,x',y'\in Q_n}
\end{equation*}

\begin{thm}       
\label{thm41}
If $\omega =\omega _1\omega _2\cdots\omega _n\in\Omega _n$, then
\begin{align*}
\varbigtriangleup _\omega ^n
  &=\sum _{j=1}^n\ket{e_{\omega _j}^n}\paren{\bra{e_{\omega _j}^n}-\bra{e_{\omega _{j-1}}^n}}
\intertext{and}
\varbigtriangledown _\omega ^n&=\sum _{j=1}^n\ket{e_{\omega _j}^n}
  \sqbrac{\mu _n(\omega _{j-1})\bra{e_{\omega _j}^n}-\mu _n(\omega _j)\bra{e_{\omega _{j-1}}^n}}
\end{align*}
where we use the conventions $\mu _n(\omega _0)=e_{\omega _0}^n=e_{\omega _{n+1}}^n=0$.
\end{thm}
\begin{proof}
We first observe that
\begin{equation*} 
\sum _{j=1}^n\ket{e_{\omega _j}^n}\paren{\bra{e_{\omega _n}^n}-\bra{e_{\omega _{j-1}}^n}}e_{\omega _k}^n
  =e_{\omega _k}^n-e_{\omega _{k+1}}^n
\end{equation*}
On the other hand
\begin{equation}         
\label{eq41}
\varbigtriangleup _\omega ^ne_{\omega _k}^n(x)
  =\sqbrac{e_{\omega _k}^n(x)-e_{\omega _k}^n(\omega _{\ab{x}-1})}\delta _{x,\omega _{\ab{x}}}
\end{equation}
Now the right side of \eqref{eq41} is zero if $\omega _{\ab{x}}\ne x$, $1$ if $\omega _k=x$ and $-1$ if $\omega _{k+1}=x$. The first result now follows. The second result is similar.
\end{proof}
The proof of the next theorem is similar to that of Theorem~\ref{thm41}.

\begin{thm}       
\label{thm42}
If $\omega =\omega _1\omega _2\cdots\omega _n$, $\omega '=\omega '_1\omega '_2\cdots\omega '_n\in\Omega _n$,
then
\begin{align*}
\varbigtriangleup &_{\omega ,\omega '}^n
  =\sum _{j,k=1}^n\ket{e_{\omega _j}^n\otimes e_{\omega '_k}^n}
  \sqbrac{\bra{e_{\omega _j}^n\otimes e _{\omega '_k}^n}-\bra{e_{\omega _{j-1}}^n\otimes e_{\omega '_{k-1}}^n}}
\intertext{and}
\varbigtriangledown &_{\omega ,\omega '}^n\\
  &\!=\!\sum _{j,k=1}^n\ket{e_{\omega _j}^n\otimes e_{\omega '_k}^n}
  \sqbrac{D_n(\omega _{j-1},\omega '_{k-1})\bra{e_{\omega _j}^n\otimes e_{\omega '_k}^n}
  \!-\!D_n(\omega _j,\omega '_k)\bra{e_{\omega _{j-1}}^n\otimes e_{\omega '_{k-1}}^n}}
\end{align*}
\end{thm}
It follows from Theorem~\ref{thm42} that
\begin{align}         
\label{eq42}
\rscript _{\omega ,\omega '}^n
  &=\varbigtriangledown_{\omega ,\omega '}^n-\varbigtriangledown _{\omega ',\omega}^n\notag\\
  &=\sum _{j,k=1}^n\left[D_n(\omega _{j-1},\omega '_{k-1})
  \ket{e_{\omega _j}^n\otimes e_{\omega '_{k-1}}^n}\bra{e_{\omega _j}^n\otimes e_{\omega '_k}^n}\right.\notag\\
  &\quad \left.-D(\omega '_{j-1},\omega _{k-1})
  \ket{e_{\omega '_j}^n\otimes e_{\omega _k}^n}\bra{e_{\omega '_j}^n\otimes e_{\omega _k}^n}\right]\notag\\
  &\quad +\sum _{j,k=1}^n\left[D_n(\omega '_k,\omega _j)
  \ket{e_{\omega '_j}^n\otimes e_{\omega _k}^n}\bra{e_{\omega '_{j-1}}^n\otimes e_{\omega _{k-1}}^n}\right.\notag\\
  &\quad \left.- D(\omega _j,\omega '_k)
  \ket{e_{\omega _j}^n\otimes e_{\omega '_k}^n}\bra{e_{\omega _{j-1}}^n\otimes e_{\omega '_{k-1}}^n}\right]
\end{align}
The matrix element representations of $\dscript _{\omega ,\omega '}^n$ and $\tscript _{\omega ,\omega '}^n$ can now be obtained from \eqref{eq42}

\end{document}